\documentclass[12pt]{article}
\usepackage{arxiv}

\usepackage{algorithm}
\usepackage{algorithmic}
\usepackage{todonotes}
\usepackage{mathtools}
\usepackage{subcaption}
\usepackage{amsthm}
\DeclareMathOperator*{\argmin}{arg\,min}
\usepackage{amsmath,amsfonts}
\usepackage{url}
\usepackage{multirow}

\usepackage{booktabs}
\newtheorem{theorem}{Theorem}
\usepackage{secdot}

\makeatletter
\let\latexparagraph\paragraph
\RenewDocumentCommand{\paragraph}{som}{%
  \IfBooleanTF{#1}
    {\latexparagraph*{#3}}
    {\IfNoValueTF{#2}
       {\latexparagraph{\maybe@addperiod{#3}}}
       {\latexparagraph[#2]{\maybe@addperiod{#3}}}%
  }%
}
\newcommand{\maybe@addperiod}[1]{%
  #1\@addpunct{.}%
}
\makeatother 

\title{EEvA: Fast Expert-Based Algorithms for Buffer Page Replacement}

\author{Alexander Demin \\ Ershov Institute of Informatics Systems \\ Novosibirsk, Russia \\ \texttt{alexandredemin@yandex.ru} \And 
Yuriy Dorn \\ 
MSU Institute for Artificial Intelligence \\ Moscow Institute of Physics and Technology \\
Moscow, Russia \\
\texttt{dornyv@my.msu.ru} 
\AND
Aleksandr Katrutsa \\
Skoltech, AIRI \\
Moscow, Russia \\
\texttt{amkatrutsa@gmail.com}
\And
Daniil Kazantsev \\
MSU Institute for Artificial Intelligence \\
Moscow, Russia \\
\texttt{dangl11626@gmail.com} 
\AND 
Ilgam Latypov\\
MSU Institute for Artificial Intelligence \\
Moscow Institute of Physics and Technology \\
Moscow, Russia \\
\texttt{i.latypov@iai.msu.ru} 
\And 
Yulia Maximlyuk \\
Sobolev Institute of Mathematics \\
Novosibirsk, Russia\\
\texttt{yumaximlyuk@gmail.com}
\And 
Denis Ponomaryov \\
Ershov Institute of Informatics Systems \\
Novosibirsk, Russia \\
\texttt{ponom@iis.nsk.su}
}
\date{}

\begin{document}

\maketitle



\begin{abstract}
Optimal page replacement is an important problem in efficient buffer management. The range of replacement strategies known in the literature varies from simple but efficient FIFO-based algorithms to more accurate but potentially costly methods tailored to specific data access patterns. The principal issue in adopting a pattern-specific replacement logic in a DB buffer manager is to guarantee non-degradation in general high-load regimes. In this paper, we propose a new family of page replacement algorithms for DB buffer manager which demonstrate a superior performance wrt competitors on custom data access patterns and imply a low computational overhead on TPC-C. We provide theoretical foundations and an extensive experimental study on the proposed algorithms which covers synthetic benchmarks and an implementation in an open-source DB kernel evaluated on TPC-C.

\end{abstract}

\maketitle

\section{Introduction}
A buffer manager is a critical component of database management systems (DBMS). 
It helps to smooth the impact of I/O speed on query execution latency. 
A typical task of a buffer manager is to promote/demote replacement of pages in the buffer to serve a workload most efficiently. 
This is implemented by an eviction strategy that identifies victim pages. 
The strategy must be adaptive or general enough to account for different data access patterns.


The principal challenge in adopting novel replacement strategies in a DBMS is related to the severe performance requirements to the buffer manager.
The computational overhead of a replacement strategy must be kept very small to avoid a cumulative negative effect on workload latency. 
Also, as I/O becomes cheaper on new storage devices (beginning from SSD), the benefit of highly accurate (in terms of the buffer hit ratio) but computationally expensive page replacement strategies (in terms of the total query latency) becomes limited. 
The trade-off between the accuracy and computational overhead is often resolved in favor of heuristics-based policies, which behave reasonably well only for some typical data access patterns. 
Many such heuristics are based on certain properties (e.g., temporal or spatial) of requested pages.

The problem of designing an optimal eviction strategy for a DBMS is similar to the CPU cache management problem in which the operation system calls to the memory play the same role as queries to the database.
In CPU caching, a page is a natural object of reasoning, and due to program variability, it is hard to consider higher-level objects to be used in heuristics. 
In contrast, in DB buffer management, the range of programs (i.e., query execution plans) and subroutines (physical operators) is limited. 
Thus, different types of data granularity can be considered for the design of replacement policies.  
Plan operators (such as, e.g., sequential or index scan) define typical page request patterns, while sequences of repeated queries determine patterns of accessed tables. 
Based on this observation, we argue in this paper that novel page replacement strategies for a DB buffer manager can be developed with the help of lightweight models that capture data access patterns on different levels of granularity, including page, operator, and table/query levels. The main contributions of our paper are as follows.
\begin{itemize}
    \item Based on the framework of expert-based algorithms, we propose a novel expert-based page replacement algorithm (EEvA) for DB buffer manager.
    We provide heuristic-based instances of EEvA, namely EEvA-Greedy, EEvA-Seq, and EEvA-T, for different application scenarios. 
    \item We show a relation between optimal page replacement and online convex optimization by representing changes in buffer states as a specific Markov Decision Process.
    As an application of this result, we establish regret bounds for EEvA algorithm.
    \item We implement EEvA-based algorithms in a synthetic experimental environment to emulate different access patterns and show that they outperform relevant competitors regarding hit ratio and latency. Also, we implement an instance of EEvA in an open-source DB engine and show that it provides better hit rates and higher transaction counts on the TPC-C benchmark.
\end{itemize}

\section{Related works}\label{sec::related}
We begin our exposition with a brief summary on the main types of replacement strategies used for memory management and highlight their benefits and limitations.

\subsection{FIFO-based strategies}
\label{subsec::fifo}
In the well-known FIFO-based approach~\cite{yang2023fifo}, a victim page is found based on the FIFO order of processed pages. A quantitative method for FIFO cache analysis is provided in \cite{guan2013fifo}.
QD-LP algorithm~\cite{yang2023fifo_better_lru} extends this method and suggests enriching eviction strategies with lazy promotion and quick demotion heuristics.
The authors prove the importance of such heuristics with an extensive numerical simulation based on long query traces from real databases. 
The basic version of FIFO algorithm is easy to implement and affordable from the computational point of view, but it does not adapt to data access patterns.

\subsection{Recency-based strategies}
\label{subsec::LRU}
In this approach, the decision to label a page as a victim is made based on the analysis of page usage recency.
The well-known implementation of this approach is LRU algorithm~\cite{touzeau2019fast,mattson1970evaluation}, which labels the least recently used page as a victim.
Slightly modified versions of LRU  add flexibility via an early eviction strategy (EELRU algorithm~\cite{smaragdakis1999eelru}) and hierarchical processing of clean and dirty pages (CFLRU~\cite{park2006cflru} and CFDC~\cite{ou2009cfdc} algorithms). 
The basic version of LRU algorithm is easy to implement and affordable for medium-sized DBs.
However, it is not adaptive, can perform badly even in very simple scenarios, and can not differentiate data access types. 

\subsection{Frequency-based strategies}
\label{subsec::LFU}
Another well-known approach is frequency-based.
Classical algorithms here are LRU-k~\cite{o1993lru} and LFU~\cite{karedla1994caching,effelsberg1984principles}.
These algorithms determine victim pages based on the simple frequency statistics of page utilization.
To avoid the overhead of LRU-k, in 2Q algorithm~\cite{johnson19942q}, two queues based on the frequency of page requests were used.
The recent modification of LFU is tinyLFU algorithm~\cite{einziger2017tinylfu}, which addresses issues of a skewed access distribution and builds upon Bloom filter theory.
However, this family of algorithms shares the same strong and weak points with LRU.

\subsection{Randomized strategies}
\label{subsec::rand}
Randomized approaches have been used to improve the performance of eviction algorithms.
For example, OnlineMIN~\cite{brodal2015onlinemin} has $\mathcal{O}(\log k)$ worst-case time complexity, where~$k$ is a cache size.
MARKER algorithm~\cite{fiat1991competitive} introduces randomization in LRU to achieve logarithmic competitiveness with respect to the optimal algorithm. 
Other approaches~\cite{mcgeoch1991strongly,bansal2012primal, coester2022competitive} offer promising theoretical guarantees, but their performance can vary significantly in real-world scenarios due to the inherent unpredictability and computation overhead of these algorithms.

\subsection{Machine Learning-based strategies}
\label{subsec::ML}
Recently, many ML-based cache and buffer management strategies have been proposed (see, for example, the recent survey~\cite{liu2022survey}). 
The main idea behind these approaches is to equip a buffer manager with an ML oracle that helps to learn data patterns to achieve desirable metrics, such as page hit rate or latency~\cite{boyar2017online}. 
In~\cite{lykouris2021competitive}, the authors propose an online algorithm with a black-box ML predictor for caching.
A construction of the optimal oracle has been first proposed in~\cite{belady1966study}, where a theoretical analysis has been presented.
This optimal oracle exploits information about future inputs, so it has mostly a theoretical value.
However, in~\cite{wu2022survey}, the authors proposed a method based on gradient boosting to approximate Belady's algorithm. 
Unfortunately, the practical value of ML-based replacement strategies is limited due to the excessive computational overhead.

\subsection{Online Convex Optimization in eviction strategies}
\label{subsec::OCO}
Another important line of research considers eviction strategies as a sequential decision process and uses Online Convex Optimization (OCO) to learn an optimal eviction policy. 
Our work follows this approach. 
In~\cite{paschos2019learning}, the authors proposed a caching algorithm based on a straightforward application of online gradient descent. 
The main problem with this approach is that it requires an ML oracle (that can be computationally demanding). 
In LeCAR~\cite{vietri2018driving} and CACHEUS~\cite{rodriguez2021learning}, the authors proposed an idea to treat basic eviction strategies (e.g., LRU, LFU, etc.) as experts. They used OCO-based aggregation of expert predictions to construct composite strategies. 
The main drawback of this approach is that one has to trigger several expert algorithms in parallel upon each (page) request which can imply non-acceptable overhead in real-world scenarios.

In general, common algorithms for page replacement are more practical and computationally efficient. 
However they have a limited ability to adapt, they do not take into account data access types and patterns, and can perform poorly even in simple scenarios. Recently proposed algorithms are more adaptive, but they imply a non-negligible computational overhead. 
In this paper, we propose an approach that offers a compromise between flexibility and performance. 
We show that this approach can provide algorithms that are adaptive to data access patterns and achieve high performance even in heavy-load scenarios.




\section{Buffer management problem as online learning problem}
Our approach is based on the framework of online learning with experts. 
First, we provide the required background and then discuss how to adopt the online learning framework in developing new page replacement strategies.

\subsection{General expert-based framework}
\label{sec:prediction with expert advice}
Prediction with expert advice is a well-studied field of online learning for sequential decision-making~\cite{cesa2006prediction}. 
The general problem statement can be described as follows.
We have an agent that makes a final prediction based on predictions made by a set of $n$ experts.
The final prediction is composed by means of an auxiliary aggregation rule.
After all expert predictions are made, the ground-truth value is known, and therefore, the value of a loss function can be computed for the agent and all experts. Based on the observed loss, experts may modify internal parameters to adjust subsequent predictions, and the agent may adjust the aggregation rule.


More formally,  let $\mathcal{D}$ be an available decision set, $T$ a time horizon, $f_{i, t}$ the prediction of the $i$-th expert at time $t$ ($1 \leq t \leq T$), and let $l_t(\cdot)$ be the loss function at time $t$. 
The main idea behind the prediction with expert advice is to select a particular expert~$i$ at round~$t$ and take its prediction~$f_{i, t}$ as the agent prediction~$d_t$.
The selection of an expert at round $t$ is typically made in a probabilistic manner, i.e., the $i$-th expert has weight $w_i^t$ and the probability distribution over the experts to be selected at round~$t$ is defined as 
\[
\mathbb{P}(d_t = f_{i, t}) = \frac{w_i^t}{\sum_{j=1}^n w_j^t}.
\]
Here, we assume that experts generate different predictions.
If predictions of some experts coincide, then the expression for the probability can be adjusted.
The optimal distribution over experts is constructed by accurately tracking expert performance and adjusting the corresponding weights $w_i^t$.

To introduce a loss function that can be used for tuning weights, we consider the vector of probabilities $x_t \in \mathbb{R}^n$ such that $[x_t]_i = \frac{w_i^t}{\sum_{j=1}^n w_j^t}$.
Then the loss $\mathcal{L}_t$ at the $t$-th round can be estimated as 
\begin{equation}
\mathcal{L}_t(x_t) = \mathbb{E}_{d_t}[l_t(d_t)] = \langle x_t, g_t \rangle,
\label{eq::total_loss}
\end{equation}
where $g_t \in \mathbb{R}^n$ is a vector of expert losses, i.e., $[g_{t}]_i = l_t(f_{i,t})$.
The particular form of $l_t$ is discussed in the further sections.


%
The main performance metric used in prediction with expert advice setting is regret defined as
\[
\min_{u\in \Delta_+^n} \left \{R_T(u)\stackrel{\mathclap{\text{def}}}{=} \sum_{t=1}^T \left (l_t(x_t) - l_t(u)\right) \right \},
\]
where $\Delta_+^n = \{u\in \mathbb{R}^n \mid \sum_{i=1}^n u_i = 1, \quad u_i \geq 0 \}$ is the probability simplex.
State-of-the-art algorithms for prediction with expert advice and their modifications are based on the online mirror descent~\cite{hazan2016introduction}. 
The simplest version of online mirror descent is given in Algorithm~\ref{alg:online_md} below.

\begin{algorithm}[!ht]
\caption{ Online mirror descent  }\label{alg:online_md}
\begin{algorithmic}[1]
\REQUIRE{learning rate $\mu > 0$, time horizon $T$}
    \STATE Initialize experts' weights $w_i^0=1$, for all $i = 1, \dots, n$
    \FOR{$t = 1, \ldots, T$}  
        \STATE Sample an expert $j$ according to the distribution \[\mathbb{P}(d_t = f_{j, t}) = \frac{w_j^t}{\sum_{i=1}^n w_i^t}.\]
        \STATE Set the agent prediction $d_t = f_{j,t}$ to the prediction of the $j$-th expert.
        \STATE Compute loss vector $g_t$
        \STATE Update the weights for experts based on the current loss:
        \[
        w_i^{t+1} = w_i^t \cdot \exp[\mu \cdot [g_{t}]_i], \quad i=1,\ldots,n
        \]
    \ENDFOR
\end{algorithmic}
\end{algorithm}


This algorithm converges with the optimal regret rate $O\left(\sqrt{T\log n}\right)$ under mild assumptions~\cite{hazan2016introduction}. 
Online mirror descent is closely related to the follow-the-perturbed leader~\cite{cohen2015following} and can be seen as a randomized version of the greedy expert selection strategy.

Prediction with expert advice can be viewed as a method to construct a single strong predictor based on (an aggregation of) a set of weak predictors. 
Therefore, selecting experts is one of the most important factors directly affecting performance. 
For example, caching algorithms like LeCAR~\cite{vietri2018driving} and CACHEUS~\cite{rodriguez2021learning} can be given within the experts' advice framework, in which the experts are LFU, LRU, and other simple heuristics, respectively. 
Such experts are called \emph{dynamic} since they tune predictions according to environmental changes.
To obtain a higher gain from aggregation, one usually needs to have a diverse set of experts that make different predictions. 
On the other hand, the more dynamic experts are used for aggregation, the more computational resources are required in each decision round.

Alternatively, one can use \textit{static} experts that are not adaptive to the environment. They are weaker predictors, but they are cheaper to maintain. 
In the next section, we discuss which static experts can be used to build fast and accurate page replacement algorithms.

\subsection{Experts for page replacement}
\label{sec::experts_db_buffer}
One can consider experts corresponding to different levels of data access granularity. 
For example, one can naturally consider every
\emph{page} as an expert that votes to keep it in the buffer.   
However, this approach faces the obvious problem that pages are mutable. 
If the page content changes, the aggregated statistics based on the previous content can not be directly used. 
Moreover, when evicted from the buffer, a page can be subject to change by a background modification thread (e.g., \emph{vacuum}) of the storage manager. 
Thus, we propose to consider a page based expert as active until the page is evicted from the buffer. 
This requires a special initialization procedure for page weight when placing a page into the buffer. We implement this feature by means of higher-grained experts which correspond to data tables. 


The single purpose of table-based experts is to provide support for the initialization of weights of page-based experts. 
We emphasize that table-based experts do not directly participate in the voting for the aggregated agent's decision. 

Also, taking tables into account allows us to distinguish between types of data requests. 
For example, a request may touch pages from the same or different tables; it can be an index or sequential scan of a particular table.
Typically, a request to get a few tuples from a table is executed by means of an index scan; further in the paper, we refer to these requests as \texttt{get}-type queries. 
The cost of processing \texttt{get}-queries is relatively low since searching by an index is based on the internal auxiliary data structure built in advance~\cite{gani2016survey}.
However, if a request touches many tuples, the cost of searching by an index can be rather high. In this case, a sequential scan is preferred and
even loading redundant pages may not induce a significant overhead. We further refer requests of this kind as \texttt{scan}-type queries.  
We argue that it is important to distinguish between index and sequential data access when building page replacement strategies. 

In the next section, we describe our expert-based framework for the development of page replacement strategies. We formalize the above-introduced experts, describe the corresponding initialization procedure, and show how \texttt{get}-type and \texttt{scan}-type queries can be reflected in the computation of 
experts' weights. 

\section{Buffer Eviction Strategy Based on Expert Model}

We base our framework on the general online optimization setting and adapt the online mirror descent algorithm (see Algorithm~\ref{alg:online_md}) to build a flexible eviction algorithm using static experts.
To show the convergence bound for our algorithm, we introduce a complex but accurate Markov Decision Process (MDP) model of page replacement.

\subsection{Markov Decision Process model for page replacement}
\label{sec:mdp_}

Markov Decision Process (MDP) is a framework to model sequential decision-making in stochastic environments~\cite{feinberg2012handbook}. 
MDP is defined as a tuple $(\mathcal{S}, \mathcal{A}, \mathbf{P}, d_1, \{\mathbf{r}\})$, where $\mathcal{S}$ is a state space, $\mathcal{A}$ an action space, $\mathbf{P}$ a transition matrix (which specifies for each pair of states $s$ and $s'$ and each action $a$ the probability of arriving from state $s$ to state $s'$ after choosing action $a$), an initial distribution $d_1$ over $\mathcal{S}$, and a sequence of reward functions $r_1, r_2, \dots$ .

We may consider buffer states (content) to be the state space $\mathcal{S}$ of MDP. 
Eviction candidates (victim pages) together with the distinguished $\{0\}$ element (no eviction needed) form the action space $\mathcal{A}$. 
Then, the transition matrix is 
\[
\mathbf{P}(\mathcal{B}' \mid \mathcal{B}, a) = \begin{cases}
\sum_{e\in \mathcal{B}} p_e, \quad \mathcal{B}' = \mathcal{B}, a = 0\\
p_e, \quad \mathcal{B}' = \{e\} \cup \mathcal{B} \setminus \{a\}, \quad e \notin \mathcal{B}, \quad a\neq 0\\
0, \quad \text{otherwise}
\end{cases}
\]
where $p_e$ is the probability that page $e$ is queried.
The initial state of the buffer is always empty.

Note that in the page replacement problem, we do not have full control of the next buffer state. 
Indeed, before we choose an action, we receive a request to a page, which must be added to the buffer if it is currently not there. 
We only control eviction. 
On the other hand, if the queried page is already in the buffer, no action is needed. 
Thus, the reward can be modeled as 
\[
r_t (\mathcal{B}, a) = \begin{cases} 
1, \quad a = 0, e \in \mathcal{B}\\
0, \quad \text{otherwise.} \end{cases}
\]
A more detailed reward model will be defined in section~\ref{subsec:upd}.

\subsection{From DB buffer model to expert-based eviction algorithms}

We now formally describe experts and the specification of the online mirror descent algorithm appropriate for building a replacement strategy.
To simplify our notation, we replace experts' weights with the corresponding experts' rewards.
The only difference here is that weights are updated in a multiplicative way, while rewards are updated in an additive way.
Moreover, the loss of every expert $l_t$ from~\eqref{eq::total_loss} equals the negative corresponding reward.
Thus, the weights in Algorithm~\ref{alg:online_md} are updated via the exponent of the learning rate multiplied by the corresponding cumulative rewards.
More formal details are discussed in subsection~\ref{sec::experts_losses_db}.


Further, we use the following notation.
Let $\mathcal{T}$ be a set of DB tables.
Assume that the $i$-th table $T_i \in \mathcal{T}$ consists of $P_i$ pages.
Also, denote by $T(j)$ the table that contains the $j$-th page.
Since pages are considered as the main experts in our algorithm, we assign a cumulative reward $r_j > 0$ to the $j$-th page, for all $j=1,\ldots, P$, where $P = \sum_{T_i \in \mathcal{T}} P_i$ is the total number of pages.
Following our proposal from Section~\ref{sec::experts_db_buffer} on using table-based experts as support for page experts, we assign rewards $v_i, i=1,\ldots,|\mathcal{T}|$ to tables for tracking of \texttt{scan}-type queries. 
Let us also denote by $P_T(\mathcal{B}_t)$ the number of pages from table $T$ residing in the buffer at round $t$.

The buffer manager processes the input set of requested pages $\mathcal{Q}_t = \{ i_1, \ldots, i_{m_t} \}$ in the $t$-th round as follows.
Pages from $\mathcal{Q}_t$ are processed sequentially.
If the $i$-th page is already in the buffer, the rewards $r_i$ and $v_{T(i)}$ are updated according to the reward update procedure presented below.
If the $i$-th page is not in the buffer, the buffer manager finds a victim page from the buffer, replaces it with the requested $i$-th page, and updates rewards~$r_i$ and $v_{T(i)}$.
This pipeline is summarized in Algorithm~\ref{alg::pipeline}. 

\begin{algorithm}[!ht]
\caption{Buffer manager pipeline}
\label{alg::pipeline}
\begin{algorithmic}[1]
\REQUIRE{a set of tables $\mathcal{T}$, a set $\mathcal{B}_t$ of pages stored in the buffer at round $t$ }
    \FOR{$t = 1, 2, \ldots$}  
        \STATE Let the $i$-th page be requested
            \IF {$i \not\in \mathcal{B}_t$}
            \STATE Use a replacement strategy to select a victim page $e \in \mathcal{B}_t$
            \STATE Replace page $e$ by the requested $i$-th page
            \STATE Reset reward of the victim page: $r_{e} := 0$
            \ENDIF
            \STATE Update rewards of the $i$-th page and the $T(i)$-th table
    \ENDFOR
\end{algorithmic}
\end{algorithm}

In the following, we provide a detailed description of the proposed algorithms for page replacement and reward update.

\subsubsection{Experts and losses}
\label{sec::experts_losses_db}
As discussed in Section~\ref{sec::experts_db_buffer}, we consider pages as primary experts and tables as support experts that are used for the initialization of page rewards.
We show how the approach is aligned with the online mirror descent framework.


A straightforward way to model replacement via experts is to use a \textit{residual buffer} after a page $e$ is evicted: $\mathcal{B}_{-e} = \mathcal{B} \setminus \{e\}$. 
Let $\mathcal{B}$ be the current buffer state, $e$ a victim page, $L_t(e)$ an accumulated reward  (e.g., hit rate or other metric) for page $e$, and let $L_t(\mathcal{B}) = \sum_{i \in \mathcal{B}} L_t(i)$ be the sum of accumulated rewards for pages in buffer $\mathcal{B}$. 
Accumulated rewards for the buffer serve as a proxy for state values in the corresponding MDP process.

The standard policy to solve the corresponding MDP is to maximize the next state value. 
If we consider residual buffer $\mathcal{B}_{e}$ after the eviction of page $e$ as an expert, this greedy policy coincides with tracking the best expert. 
Online Mirror Decent algorithm given in Section~\ref{sec:prediction with expert advice} solves a regularized version of this problem, which is known as 'follow the perturbed leader'. 
To implement this, the algorithm samples an expert with probability $p_{\mathcal{B}_e} \sim w_{\mathcal{B}_e} = w_0 \cdot \exp \{\mu \cdot (L(\mathcal{B}) - L(e)) \} \sim w_0 \cdot \exp \{- \mu \cdot L(e) \}$. 
Hence, the algorithm provides the same result as if we used page $e$ as an expert with weight $w_e = w_0 \cdot \exp \{- \mu \cdot L(e) \}$ to represent the residual buffer.



\subsubsection{Update rule for rewards}
\label{subsec:upd}
An important ingredient in identifying a victim page is the reward update procedure.
The main features of this procedure are the proper initialization of page rewards and an update scheme based on query frequencies and query types. Page rewards are updated differently for \texttt{scan}- and \texttt{get}-types of queries as formalized by update coefficients $\alpha > 0$ and $\beta > 0$ in Algorithm~\ref{alg::update_scheme} below.
If the query is \texttt{get}-type, then $\alpha$ is used to update the page and table rewards, otherwise $\beta$ is used.
The interpretation of the introduced coefficients $\alpha$ and $\beta$ is explained further in section~\ref{sec::cost_model}.

Then we can model the reward (introduced in section~\ref{sec:mdp_}) as 
\[
r_t (\mathcal{B}_t, q_t) = \sum_{e \in \mathcal{B}} r_t(e, q_t) \mathbb{I}[q_t \in \mathcal{B}_t],
\]
where $q_t$ is a page requested at round $t$ by a query $\mathcal{Q}$ 
\[
r_t(e, q_t) = \begin{cases} 
\alpha, \quad e=q_t \in \mathcal{B}_t, \text{$\mathcal{Q}$ is of \textit{get} type},\\
\beta, \quad e=q_t \in \mathcal{B}_t, \text{$\mathcal{Q}$ is of  \textit{scan} type},\\
0, \quad \text{otherwise.} \end{cases}
\]
Then the accumulated reward for a page $e$ is $L_t(e) = \sum_{s=1}^t r_t(e, q_s)$. 
\begin{algorithm}[!ht]
\caption{Reward update procedure for given query}
\label{alg::update_scheme}
\begin{algorithmic}[1]
\REQUIRE{a set of pages from buffer $\mathcal{B}$, update coefficients $\alpha > 0$ (for \texttt{get}-type query) and $\beta > 0$ (for \texttt{scan}-type query), rewards  $r_j$ and $v_i$ for pages and tables, respectively, integer index~$t$, a requested page index~$i_j$, the type of page request $R$ (\texttt{get}/\texttt{scan}-type query), the total number of DB tables $|\mathcal{T}|$}
    \IF{$t = 1$}
    \STATE Initialize tables' rewards as $v_i := 0$, for $i=1,\ldots,|\mathcal{T}|$.
    \ENDIF
        \IF {$R$ is \texttt{get}-type}
        \STATE set $\delta := \alpha$
        \ELSE
        \STATE set $\delta := \beta$
        \ENDIF
        \IF {$i_j \in \mathcal{B}$}
            \STATE Update reward of the requested page as $r_{i_j} := r_{i_j} + \delta$
            \ELSE
            \STATE Initialize page reward $r_{i_j} := v_{T(i_j)}$.
            \STATE Update page reward $r_{i_j} := r_{i_j} + \delta$
        \ENDIF
        \IF{$R$ is \texttt{get}-type}
        \STATE Update the reward of the table containing the queried page: $v_{T(i_j)} := v_{T(i_j)} + \frac{\delta}{P_{T(i_j)}(\mathcal{B})}$
        \ELSE 
        \STATE Update the reward of the table containing the queried page: $v_{T(i_j)} := v_{T(i_j)} + \delta$
        \ENDIF
\end{algorithmic}
\end{algorithm}

\subsubsection{Queries cost modeling}
\label{sec::cost_model}
In the previous section, we introduced coefficients $\alpha$ and $\beta$ to update page rewards on \texttt{get}-type and \texttt{scan}-type queries, respectively.
We now show how these coefficients relate to the cost of the corresponding type of query.
Note that
\begin{equation*}
    r_i^t = r_{initial} + \alpha \sum_{s=1}^t \mathbb{I}_{(idx(s) = i)}+\beta \sum_{s=1}^t \mathbb{I}_{(scan(s) = i)}
\end{equation*}
and hence,
\begin{equation*}
    \lim_{t\rightarrow \infty}\frac{r_i^t}{t} = \alpha \mathbb{P}[idx = i] + \beta \mathbb{P}[scan = i],
\end{equation*}
where $\mathbb{P}[idx = i]$ and $\mathbb{P}[scan = i]$ are the probabilities of events that page $i$ is  requested by \texttt{get}-type and \texttt{scan}-type query, respectively.
Thus, if $\alpha$ and $\beta$ accurately represent the costs of \texttt{get}-type and \texttt{scan}-type queries, then $r_i^t$ represents the cumulative reward that we get when keeping the page $i$ in the buffer till round $t$.
Hence, $\frac{r_i^t}{t}$ can be viewed as the average reward (saved cost) produced by the page that stays in the buffer.



\subsubsection{Estimation of request costs}
To properly initialize coefficients $\alpha$ and $\beta$ from the reward update procedure, we need to estimate the costs of the considered types of queries.
If we assume that a B-tree index search is used for a \texttt{get}-type query, then we can estimate the cost of such query as $\mathcal{O}(\log P_i)$, where $P_i$ is the number of pages in the $i$-th table.

Now consider a \texttt{scan}-type query of size $r$, i.e. we take $r$ pages from the storage, starting from the $j$-th page of the $i$-th table.
Then, the search for the $j$-th page costs $O(\log(P_i))$ and the cost of loading the remaining pages is $O(r)$.
If only $a$ pages are required, then
a \texttt{scan}-type query can be cheaper than the \texttt{get}-type provided $c_{idx} \cdot \log(T) + r \leq a\cdot c_{idx} \cdot \log(T),$ where $c_{idx}$ is the base cost for loading a page. In this case, it is cheaper to load additional pages instead of performing a sequence of index access operations.

To simplify the formal analysis, we make an assumption that the loading cost $c_{load}$ for every page in a table $T$ is the same and the cost of index search is $c_{idx} \cdot \log(P_T)$ for B-tree indexes. 
We also introduce false hit rate $\gamma \in [0, 1]$, which represents the probability that a page added to the buffer via the scan operation was actually not needed.
Then, the $i$-th page for a table $T$ at round $t$ induces the following cumulative cost: 
\begin{itemize}
    \item $c_{idx} \cdot \log(P_T)$ if page is requested by a \texttt{get}-type query,
    \item $\frac{c_{idx}\cdot \log(P_T)}{r} + 1 + \gamma \cdot c_{load}$ if the page is requested by a \texttt{scan}-type query of size $r$ and with probability $\gamma$ is not actually required
\end{itemize}
For the reward update procedure described in section~\ref{subsec:upd}, the estimated cost for \texttt{get}-type query and \texttt{scan}-type query should be used as the value of $\alpha$ and $\beta$, respectively.

\subsection{Expert-based eviction algorithms}

We now describe the remaining ingredient of Algorithm~\ref{alg::pipeline}, which is the strategy to find a victim page. 

\paragraph{EEvA}
Our first algorithm, named \texttt{EEvA}, is a probabilistic strategy based on the Online Mirror Descent. 
Active pages are primary experts, and table-based experts provide warm-up support to initialize page rewards.

The strategy picks a victim page $e \in \mathcal{B}$ by sampling pages from $\mathcal{B}$ according to the distribution 
\begin{equation}
    p_i = \frac{e^{-\mu\cdot w_i}}{\sum_{j \in \mathcal{B}} e^{-\mu\cdot w_j}},
    \label{eq::prob_eviction}
\end{equation}
where $\mu > 0$ is a given learning rate. 


\begin{theorem}
    Let $\mu = \sqrt{\frac{8}{T} \log T}$, $0<\beta<\alpha\leq 1$, and EEvA algorithm outputs a sequence $\{x_t\}_{t=1}^T$. Then the following inequality holds:
    \[
    \frac{\mathbb{E}[\sum_{t=1}^T \left (l_t(x_t) - \min_{u_t \in \mathcal{B}_t} l_t(u_t) \right )]}{T} \leq 2\sqrt{\frac{2\log T}{T}}.
    \]
\end{theorem}

\begin{proof}
We can only add one new page at each round. 
Thus, the total number of experts one can encounter in horizon $T$ is limited by $T$.
Denote by $a^*(\mathcal{B}) = \arg \min_{e \in \mathcal{B}} L_T(e)$ the post-hoc best eviction candidate. 
Suppose that there were no evictions prior to round~$T$. Then, the probability distribution on eviction candidates produced by our algorithm is the same as the one produced by Algorithm~\ref{alg:online_md}. 
Hence,  
\[
\mathbb{E}\left [\sum_{t=1}^T \left (l_t(x_t) - \min_{u_t \in \mathcal{B}_t} l_t(u_t) \right )\right ] \leq 
\]
\[
\mathbb{E}\left [\sum_{t=1}^T \sum_{\mathcal{B} \in \mathcal{S}} \mathbb{P} [\mathcal{B}_t = \mathcal{B}] \cdot \left (l_t(x_t) - r_t(a^*(\mathcal{B}, q_t) 
 \right )\right ]=
\]
\[
\mathbb{E}\left [\sum_{\mathcal{B} \in \mathcal{S}} \mathbb{P} [\mathcal{B}_t = \mathcal{B}] \cdot\sum_{t=1}^T  \left (l_t(x_t) - r_t(a^*(\mathcal{B}, q_t) 
 \right )\right ] \leq
\]
\[
\mathbb{E}\left [\sum_{\mathcal{B} \in \mathcal{S}} \mathbb{P} [\mathcal{B}_t = \mathcal{B}] \cdot 2\sqrt{2T \log T} \right ] \leq 2\sqrt{2T \log T}
\]
with respect to $2\sqrt{2T\log n}$ regret rate for Algorithm~\ref{alg:online_md} (Corollary 5.7. in \cite{hazan2016introduction}).
\end{proof}
This theorem states that expected costs for EEvA are not higher than for ex-post optimal eviction strategy plus $\sqrt{T \log T}$.

\paragraph{EEvA-Greedy.}
An alternative strategy is to replace sampling with selecting a page $e$, which has the smallest weight:
\begin{equation}
    e = \argmin_{j \in \mathcal{B}} w_j.
\end{equation}
This approach is similar to the greedy search; therefore, we refer to it as \texttt{EEvA-Greedy}.

    

\paragraph{EEvA-T}
The next strategy \texttt{EEvA-T} that we consider is the simplified version of \texttt{EEvA-Greedy} where only table rewards are used to identify a victim page.
Namely, the first page (from the sequence of buffer pages) is chosen as a victim, for which the original table has the smallest reward $v_i$.
For example, this strategy can be used in case of a large buffer when tracking the rewards of each active page is too expensive. 

Denote by $\mathcal{B}_t$ the buffer state (i.e., the current set of pages in the buffer) at round $t$. 
Set $\mathcal{B}_t^{-1} = \mathcal{B}_t\setminus \{\argmin_{e \in \mathcal{B}_t} w_e\}$, i.e., the set of pages from $\mathcal{B}_t$ excluding the page with minimal weight. 
Then consider the following scenario: page content does not change over time, each page is considered as an independent table, and its reward is saved even after eviction from the buffer. Assume also that there is only one type of query (\texttt{scan}-type or \texttt{get}-type).
This scenario brings the buffer page replacement problem very close to the cache management problem.
In particular, for this scenario, the following theorem holds.

\begin{theorem}
Let $\mathcal{B}_t$ be the sequence of buffer states produced by \texttt{EEvA-T} algorithm. Then, it holds that 
\[
\lim_{t\rightarrow \infty} \mathcal{B}_t^{-1} = \mathcal{B}_{opt}^{-1},
\]
where $\mathcal{B}_{opt}$ is the optimal buffer state. 
\end{theorem}

\begin{proof}
Since we control only eviction but can not control which page is added, this statement is not trivial. 
Another difficulty is due to the randomness of page additions. 

In the IID scenario with fixed probabilities, $\mathcal{B}_{opt}^{-1}$ consists of $k-1$ pages with the highest probability to be queried, where $k$ is the buffer size.
We can model the evolution of buffer states via a Markov chain with discrete time where $\mathcal{B}_t$ is the state at round $t$ and $\left \{\mathcal{B}_t^{-1}\cup\{e\} \right \}_{e}$ is the set of feasible next states. 
The transition probabilities are 
\[
\mathbb{P} \left [\mathcal{B}_{t+1} = \mathcal{B}_t^{-1}\cup\{e\} \mid \mathcal{B}_t\right ] = p_e,\] 
where $p_e$ is the probability that page $e$ is queried.

Since $t \rightarrow \infty$, for each page $e$, with $p_e>0$, there exists $t$ such that $e \in \mathcal{B}_t$, i.e., each page that could be requested would appear in the buffer. 
Moreover, for any page $e$, with $p_e>0$, we have $\lim_{t\rightarrow \infty}\frac{w_e}{t}= c_{get} \cdot p_e$. 
Thus, for a large $t$ we get $\argmin_{e\in \mathcal{B}_t} w_e = \argmin_{e\in \mathcal{B}_t} p_e$ and hence, $\lim_{t\rightarrow \infty} \mathcal{B}_t^{-1} = \mathcal{B}^{-1}_{opt}$.
\end{proof}

\paragraph{EEvA-Seq}
As we already mentioned, a significant problem for any eviction strategy is to guarantee a low level of computational overhead of the eviction logic in heavy-load scenarios. 

A typical solution is to use a worker thread (or simply, worker) that checks the buffer page by page in the search for a victim. The clock-sweep family of algorithms is typical baseline here.
In fact, one can consider an even more greedy strategy when the very first candidate victim page is replaced. Such strategies perform badly in terms of the miss-rate metric but have an advantage of extremely low computational overhead. 

Following this observation, we propose the next strategy \texttt{EEvA-Seq}, which employs page reward $r_i^t$ as the main feature for identifying victim pages and the search horizon that depends on the current load regime.

\begin{algorithm}[!ht]
\caption{\texttt{EEvA-Seq}}
\label{alg::eeva-seq}
\begin{algorithmic}[1]
\REQUIRE{Initial worker's position $pos(0)= pos(T_0)$, a policy cost estimate $c$, initial average weight estimate $\hat{w} := \frac{1}{|\mathcal{B}|} \sum_{e \in \mathcal{B}} w_e$}
    \FOR{$t = 0, 1, \ldots$}  
        \STATE Let $e_{new}$ be a page to be loaded into the buffer
        \WHILE{$\frac{1}{t}w_{pos(t)} > \frac{1}{t}\hat{w} + c$}
        \STATE $pos(t) := pos(t)+1.$
        \ENDWHILE
        \STATE Set $\hat{w} := \hat{w} + \frac{1}{|\mathcal{B}|}\left (w_{e_{new}} - w_{pos(t)}\right ).$
        \STATE Replace page at $pos(t)$ with page $e_{new}$.
    \ENDFOR
\end{algorithmic}
\end{algorithm}

As we already showed, the quantity $\frac{1}{t}r_i^t$ estimates the expected saved cost by keeping a page in the buffer. 
Essentially, the strategy checks which of the decisions is more cost-efficient: replace the currently visited page or skip it and move to the next position. 
We note that by definition, \texttt{EEvA-Seq} strategy is more cost-efficient than the greedy one.

\section{Experimental evaluation}

In this section, we report on the conducted experiments and demonstrate that the proposed approach outperforms SOTA competitors in terms of buffer miss ratio and the cost function~\eqref{eq::time_cost}.
In the first series of experiments, we use a synthetic environment described below. The source code to reproduce experiments is accessible at  \url{https://github.com/tmp2035/eeva}.

\subsection{Synthetic environment}
\label{sec: exp}

\paragraph{Experimental setting}
We generate a set of tables $\mathcal{T}$ such that each table contains $P_i$ pages, for $i=1,\ldots,|\mathcal{T}|$.
In our experiments, we have $|\mathcal{T}| = 50$ and $P_i \sim \mathcal{U}([P_{\max} / 2, P_{\max}])$, where $P_{\max}$ is the maximum number of pages per table; we take $P_{\max} = 1000$.
Our workload consists of a mixture of \texttt{get}-type and \texttt{scan}-type queries, with the probability of the latter being equal to $p_{scan} = 2\cdot 10^{-3}$ (i.e., in our workload, \texttt{scan}-type queries are rare).
We consider a sequence $Q = \{ q_1, \ldots, q_N\}$, where every $q_i$ is a collection of pages requested by an index or scan access to a table.
We denote by $SQ_i$ and $GQ_i$ the sets of pages from query $q_i$ requested via scan and index access, respectively, i.e. $q_i = SQ_i \sqcup GQ_i$.
Also, we denote by $c_{scan}$ and $c_{get}$ the costs of loading pages from \texttt{scan}-type query and \texttt{get}-type query, respectively. 
These costs are used to evaluate the averaged time cost $C$ of processing $Q$ as follows:
\begin{equation}
    C = \frac{1}{\sum_{k=1}^N |q_k|} (C_{scan} + C_{get}),
    \label{eq::time_cost}
\end{equation}
where
\begin{equation}
\begin{split}
    C_{scan} &= \sum_{i=1}^{N}\sum_{p_j \in MP_i} c_{scan} \mathbb{I}(p_j \in SQ_i), \\
    C_{get} &= \sum_{i=1}^{N}\sum_{p_j \in MP_i} c_{get}\mathbb{I}(p_j \in GQ_i),
    \end{split}
\end{equation}
where $MP_i$ is a set of missed pages from query $q_i$. 
Also, we use the indicator function $\mathbb{I}(x) = \begin{cases} 1, & x = true \\ 0, & x = false \end{cases}$.
In experiments, we use $c_{get} = 1$ and $c_{scan} = 0.8$.

To model patterns of access to tables and pages, we introduce auxiliary probability distributions for them.
Specifically, the query sequence is generated as follows.
Firstly, our emulator decides on a query type  (\texttt{scan} or \texttt{get}) according to the $\xi \sim Be(p_{scan})$.
If $\xi = 1$, then the query type is \texttt{scan}, otherwise the query type is \texttt{get}.
Secondly, a table is chosen according to the distribution $\mathcal{P}^{(T)}_{scan}$ or $\mathcal{P}^{(T)}_{get}$ depending on the query type.
In our experiments, we assign higher probabilities to the first third of the tables to be scanned.
In particular, denote by $p^{(t)}_{scan}$ the probability of the $t$-th table to be scanned.
Then, 
\[
p^{(t)}_{scan} = \begin{cases}
    \frac{10}{4|\mathcal{T}|}, & t \in \{1, \ldots,|\mathcal{T}|/3\} \\
    \frac{1}{4|\mathcal{T}|}, & t \in \{|\mathcal{T}|/3 + 1, \ldots, |\mathcal{T}|\},
\end{cases}
\]
i.e., the first third of the tables have ten times higher probabilities to be scanned.
The distribution $\mathcal{P}^{(T)}_{get}$ has a similar structure, but the probabilities of the first third of the tables are ten times smaller than the others to be requested by \texttt{get}-type queries.
If the query type is \texttt{scan}, then all pages from the selected table are requested.
If the query type is \texttt{get}, a single page is sampled from the selected table according to the Zipf law with parameter $q = 0.1$. 
The smaller $q$ is, the more uniform the distribution of pages within the selected table is.

Since the nature of our scenario is probabilistic, consider the expected number of \texttt{scan}-type queries denoted as $N_{scan}$. 
The total number of queries equals $N = 10^6$. 

In the next sections, we compare our algorithms with competitors in a number of scenarios (Table~\ref{tab::scenarios}) that simulate different load patterns.

\begin{table}[!ht]
\centering
\caption{The list of the considered scenarios with the corresponding parameters $p_{scan}$.}
    \begin{tabular}{lc}
    \toprule
     Scenario & $p_{scan}, \cdot 10^{-3}$ \\
    \midrule
    Low $N_{scan}$ & 0.6 \\
     Medium $N_{scan}$ & 1.3  \\
 High $N_{scan}$ & 1.8  \\
    \bottomrule
    \end{tabular}
    \label{tab::scenarios}
\end{table}

\subsection{Performance comparison in selected scenarios}
As competitors, we consider methods representing typical classes of strategies mentioned in section~\ref{sec::related}. 
In particular, we compare our approach with the recency-based strategy (LRU method), the frequency-based strategy (tinylfu method), the FIFO-based strategy (qdlp method), the ensemble strategy (cacheus method), and the provably optimal ML-based Belady strategy (belady method~\cite{belady1966study}).

\paragraph{Only \texttt{get}-type queries for uniformly distributed tables}

First, we consider the scenario where only \texttt{get}-type queries are simulated, and each page has the same probability of being requested.
The resulting expected average page weights are shown in Figure~\ref{fig:zero_scan_weights}. 
Further, we show how this pattern of the expected average page weights changes after \texttt{scan}-type queries appear.

\begin{figure}[!ht]
    \centering
    \includegraphics[width=0.5\linewidth]{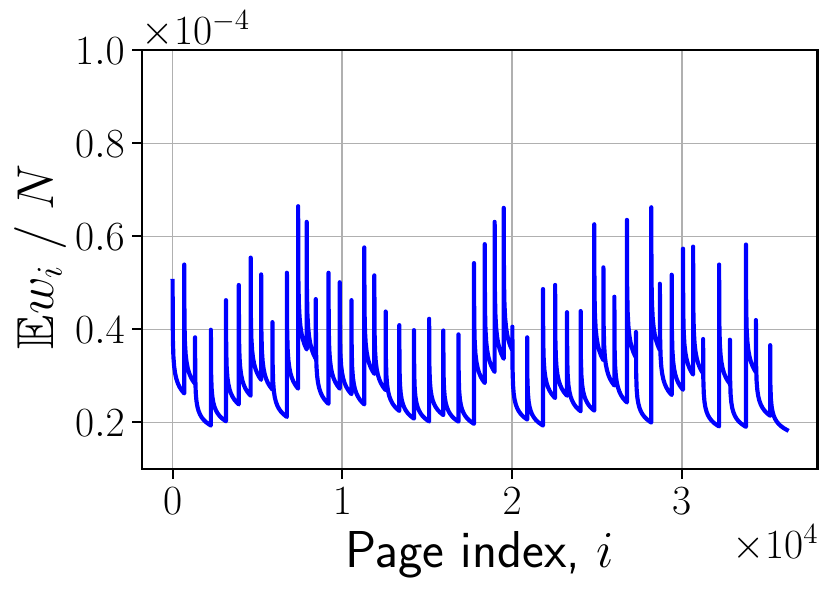}
    \caption{Expected average page weights corresponding to the scenario with only \texttt{get}-type queries, in which every page has the same probability of being requested. Weight values are spread uniformly along the interval of page indices.}
    \label{fig:zero_scan_weights}
\end{figure}

Figure~\ref{fig::zero_scan} demonstrates that EEvA-T and EEvA-Greedy algorithms show the best performance regarding the miss rate and averaged time cost (see~\eqref{eq::time_cost}).
We observe the gap in both metrics that highlights the gain from the proposed EEvA-based eviction strategies.
Although our algorithms can process \texttt{scan}-type queries properly, this experiment demonstrates high performance for the \texttt{get}-only scenario.
In the next paragraphs, we show how the expected averaged page weights and the performance of the considered algorithms are changed if the \texttt{scan}-type queries appear in the trace.

\begin{figure}[!ht]
    \centering
        \includegraphics[width=0.5\textwidth]{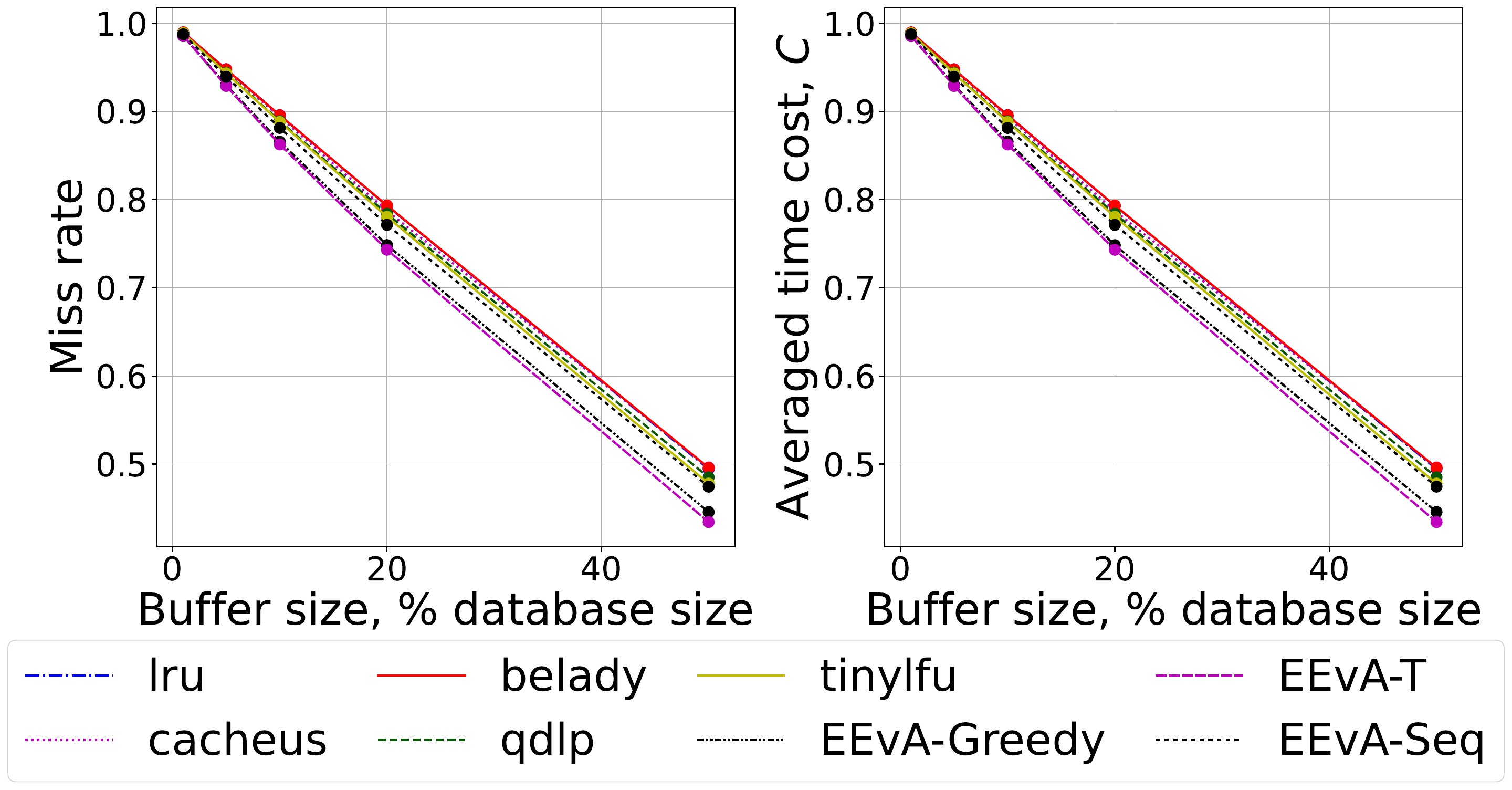}
    \caption{EEvA-T and EEvA-Greedy algorithms outperform competitors in the \texttt{get}-only scenario, where all pages are sampled for \texttt{get}-type queries with similar probabilities. These plots show the robustness of the proposed algorithms in the scenario with no scan queries.}
    \label{fig::zero_scan}
\end{figure}

\paragraph{Low number of \texttt{scan}-type queries.}

Now we consider the scenario in which the number of \texttt{scan}-type queries is low but non-zero, corresponding to $p_{scan} = 0.6 \cdot 10^{-3}$.
This value of $p_{scan}$ gives $N_{scan} = N p_{scan} = 600$.
In this case, we obtain the expected average page weights $\frac{1}{N} \mathbb{E}w_i$ given in Figure~\ref{fig:scan_lower_get_weights}.

\begin{figure}[!ht]
    \centering
    \includegraphics[width=0.5\linewidth]{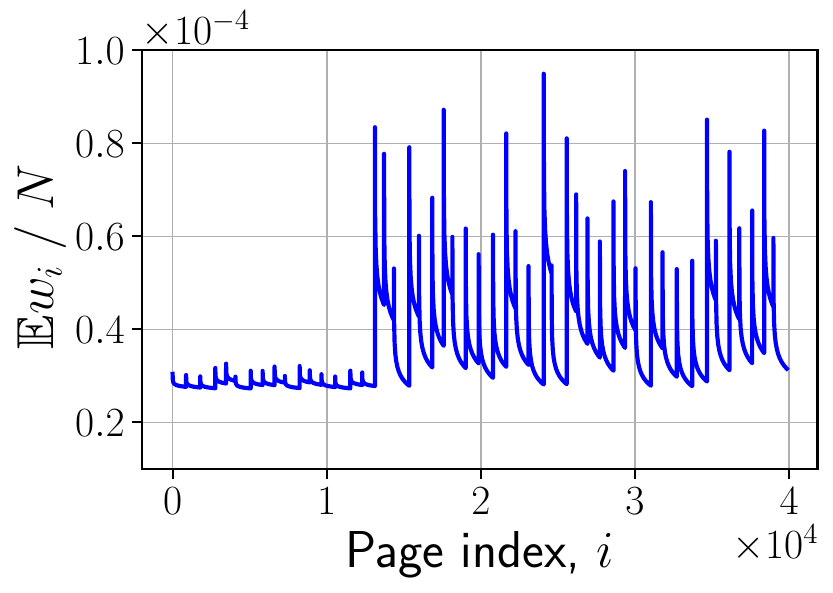}
    \caption{Expected average page weights $w_i$ show that queries to the first pages assigned to scanned tables are less costly than others. 
    Since we take into account not only the frequency of the requested pages but also the corresponding cost ($c_{scan}$ or $c_{get}$), the expected average weights accumulate this information, as well.}
    \label{fig:scan_lower_get_weights}
\end{figure}

The miss ratio and averaged time costs (see~\eqref{eq::time_cost}) obtained in this scenario are given in.
Figure~\ref{fig::scan_lower_get} shows that EEvA-T and EEvA-Greedy have the same performance and they outperform competitors wrt the both metrics.
At the same time, EEvA-Seq gives a slightly lower cost than non EEvA-based algorithms.
This result is quite promising since from our point of view, EEvA-Seq is the most appropriate candidate strategy for implementation in a DB due to its low overhead.
We provide more details in Section~\ref{sec::opengauss}.

\begin{figure}[!ht]
    \centering
        \includegraphics[width=0.5\linewidth]{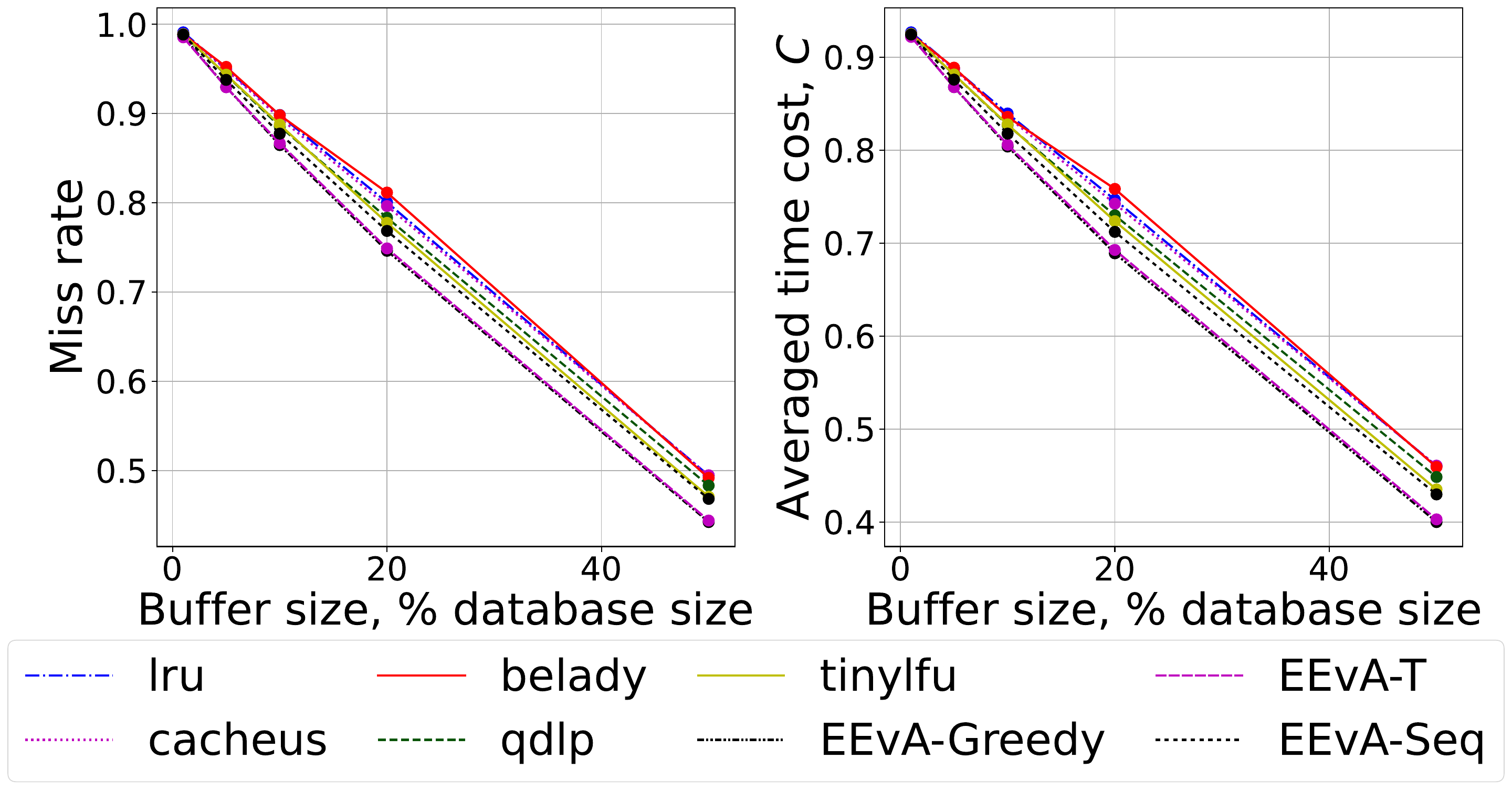}
    \caption{Comparison in the scenario with rare \texttt{scan}-type queries. The lower the miss rate and averaged time costs, the better the performance is. EEvA-based algorithms show a superior performance over the competitors wrt both metrics.}
    \label{fig::scan_lower_get}
\end{figure}

\paragraph{Medium number of \texttt{scan}-type queries.}
The next scenario considers a medium number of \texttt{scan}-type queries, with $p_{scan} = 1.3$ and $N_{scan} = Np_{scan} = 1300$. Here the numbers of \texttt{get}-type and \texttt{scan}-type queries are similar.


\begin{figure}[!ht]
    \centering
    \includegraphics[width=0.5\linewidth]{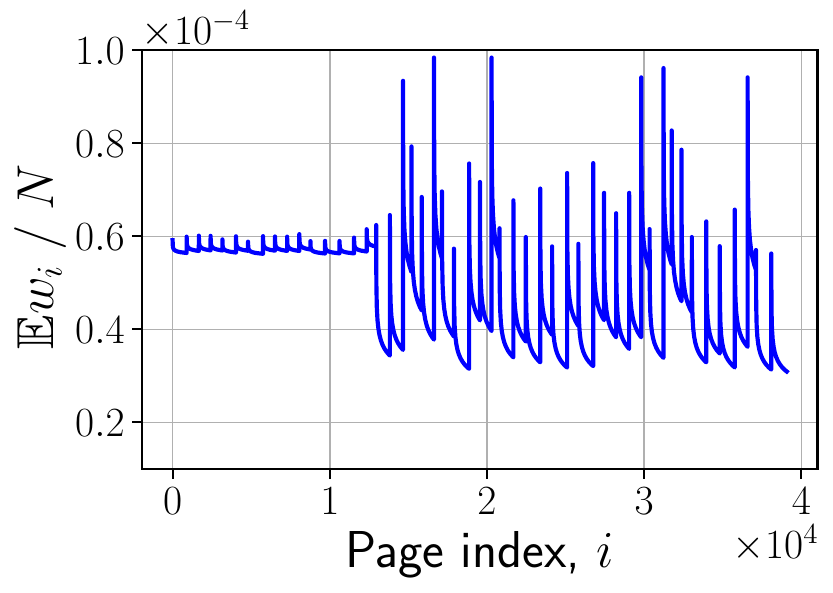}
    \caption{Expected average page weights $w_i$ show that queries to the first pages assigned to scanned tables are comparable in terms of the loading costs compared to other pages. 
    Increasing the value of $p_{scan}$ leads to increased costs of scanned pages, which affects the distribution of the shown expected average weights.}
    \label{fig:scan_eq_get_weights}
\end{figure}

Figure~\ref{fig::scan_eq_get} shows that EEvA-T and EEvA-Greedy strategies are slightly better than the alternatives in terms of both miss ratio and averaged time costs (see~\eqref{eq::time_cost}).
However, EEvA-Seq does not outperform alternatives in contrast to the previous scenario.

\begin{figure}[!ht]
\centering
    \includegraphics[width=0.5\linewidth]{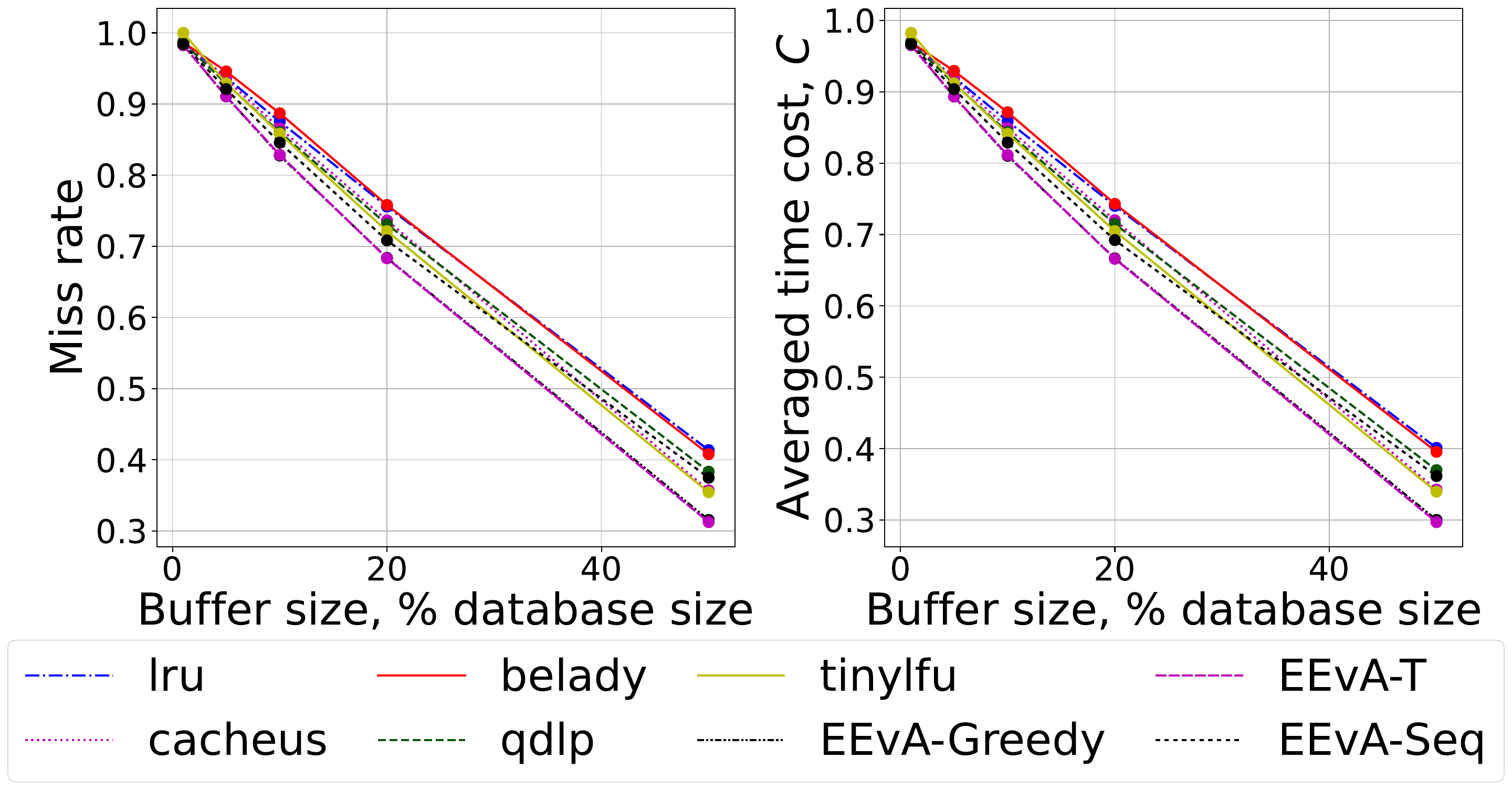}
\caption{Comparison of the considered algorithms in the scenario with a medium number of \texttt{scan}-type queries. The lower the miss rate and averaged time costs, the better performance is. EEvA-T and EEvA-Greedy strategies show a superior performance over the alternatives wrt both metrics.}
    \label{fig::scan_eq_get}
\end{figure}

\paragraph{High number of scan queries.}
Finally, we consider a scenario with a high number of \texttt{scan}-type queries: $p_{scan} = 1.8 \cdot 10^{-3}$ and $N_{scan} = Np_{scan} = 1800.$
Figure~\ref{fig::scan_greater_get_weights} shows the expected averaged weights for pages and highlights the high expected averaged weights for the first pages assigned to \texttt{scan}-type queries.

\begin{figure}[!ht]
    \centering
    \includegraphics[width=0.5\linewidth]{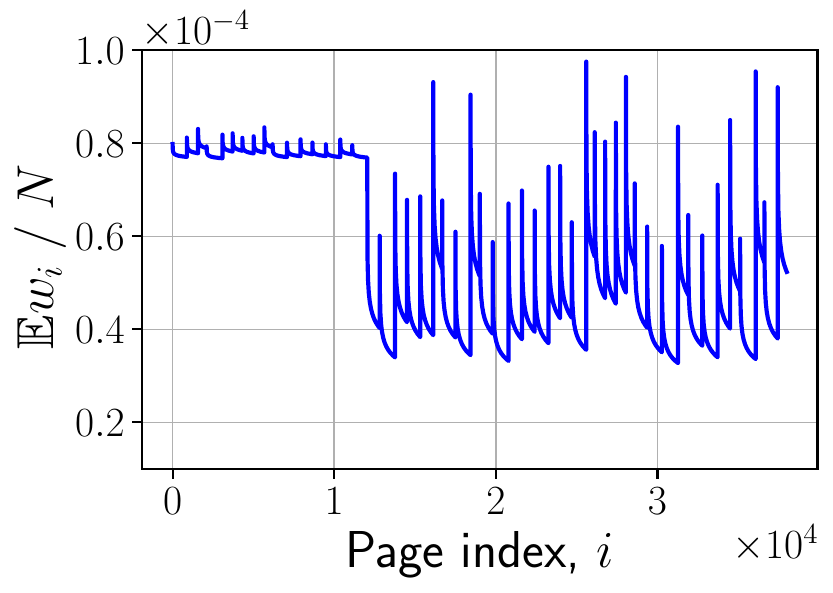}
    \caption{Expected average page weights $w_i$ show that queries to the first pages are much higher in terms of the loading costs compared to other pages.
    }
\label{fig::scan_greater_get_weights}
\end{figure}

Since EEvA-based algorithms pay special attention to \texttt{scan}-type queries, we expected that this scenario should be plausible for demonstrating the performance of our algorithms.
Indeed, Figure~\ref{fig::scan_greater_get} shows the gap in miss rate and averaged time costs~\eqref{eq::time_cost} between the lines corresponding to our methods and competitors.
Although \texttt{tinylfu} becomes closer to EEvA-T and EEvA-Greedy, it still provides a larger miss rate.
Also, in this scenario, EEvA-Seq achieves good performance compared to alternatives.

\begin{figure}[!ht]
    \centering
    \includegraphics[width=0.5\linewidth]{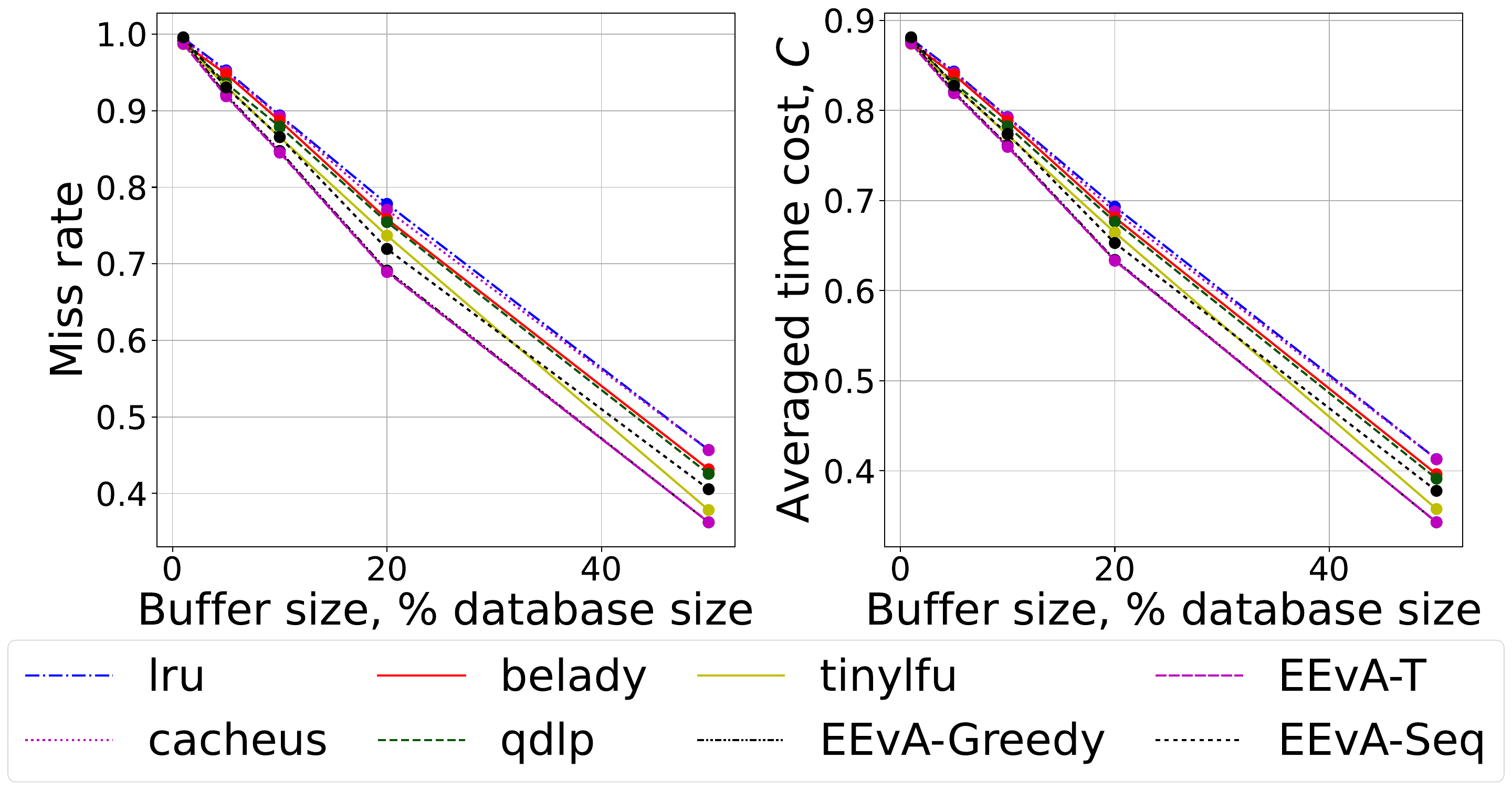}
    \caption{Comparison in the scenario with a high  number of \texttt{scan}-type queries.
    The lower the miss rate and averaged time costs, the better performance is.
    EEvA-based methods outperform almost all the considered competitors except \texttt{tinylfu}, which has a special feature to track \texttt{scan}-type queries, too.
    }
    \label{fig::scan_greater_get}
\end{figure}


\paragraph{The worst-case scenario.}
In addition to the scenarios discussed in the previous paragraphs, here we consider the worst-case scenario.
The feature of this scenario is that queries are distributed in such a way that only EEvA with sampling (see~\eqref{eq::prob_eviction}) achieves a reasonable quality and outperforms considered competitors.
In the worst-case scenario, the buffer size is smaller than the size of the entire database, and the workload is a sequence of \texttt{get}-type queries repeated 10 times.
Most of the competitor algorithms perform similarly to LRU and, therefore, give zero hit rate.
Figure~\ref{fig::worst_case} shows the barplot for different buffer sizes, where EEvA is the second-best algorithm regarding the hit rate. 
The best algorithm in this scenario is \texttt{belady}, which is provably optimal but impractical since it uses full information on future queries.


\begin{figure}[!ht]
    \centering
        \includegraphics[width=0.5\linewidth]{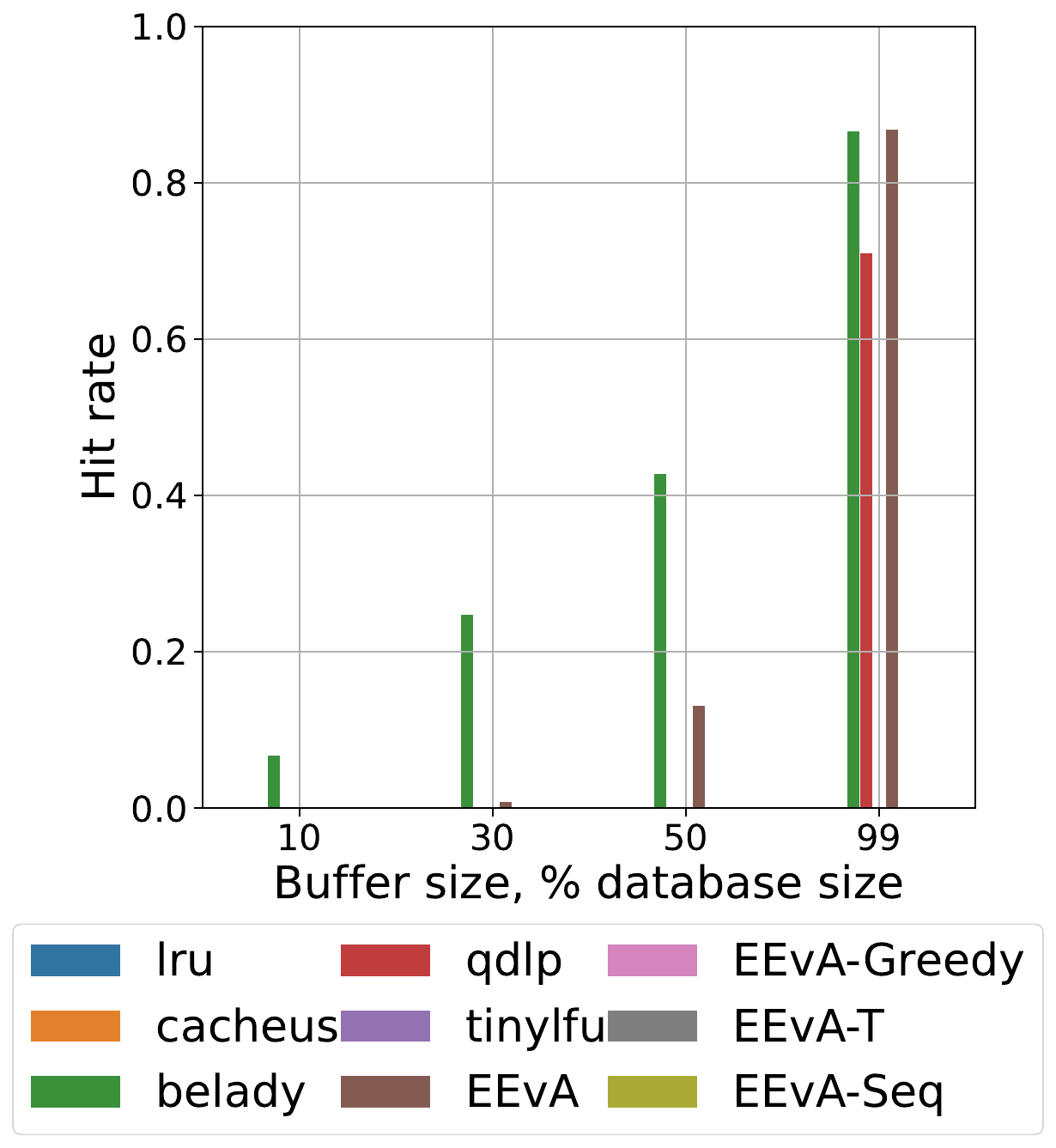}
    \caption{Comparison in the worst-case scenario for various buffer sizes. EEvA outperforms all alternatives except \texttt{belady}. Also, note that the larger the buffer size is, the higher the hit rate from EEvA, and the more similar performance of EEvA and \texttt{belady} is.}
    \label{fig::worst_case}
\end{figure}

\subsection{Dynamics of miss rate for regular scenarios}

In the previous section, we have demonstrated the miss rates and averaged time cost after processing the complete trace of queries.
However, not only the total miss rate but also the dynamics of change of the miss rate over the processed queries are crucial for evaluating the performance of the considered algorithms.
To analyze dynamics, we compute the cumulative miss rate after processing a given number of pages for every algorithm.
To improve plot readability, we present the computed cumulative miss rates with respect to the same metric for the EEvA-T algorithm.
Specifically, we show how much (in percentage) the cumulative miss rates for other algorithms are worse than the ones for the EEvA-T algorithm with an increasing number of processed pages. 
The resulting plots are given in Figure~\ref{fig::miss_dynamics}, where the averaged values and corresponding variance over 5 runs are presented.
The results demonstrate that the EEvA-T algorithm outperforms competitors: it gives the smallest cumulative miss rate after processing a moderate number of pages.
Also, EEvA-Greedy is the second-best algorithm in all the considered regular scenarios.
Note that EEvA-Greedy significantly outperforms other non EEvA-based competitors, even considering the variance of the metric.
The dynamics of the relative cumulative miss rate indicate that we achieve smaller miss rates than competitors not only at the final stage of trace processing but also after processing a moderate number of pages.
Future work will include a theoretical analysis of EEvA-based algorithms and the identification of a sufficient number of pages to obtain proper page weights.

\begin{figure}[!ht]
    \centering
        \includegraphics[width=0.7\linewidth]{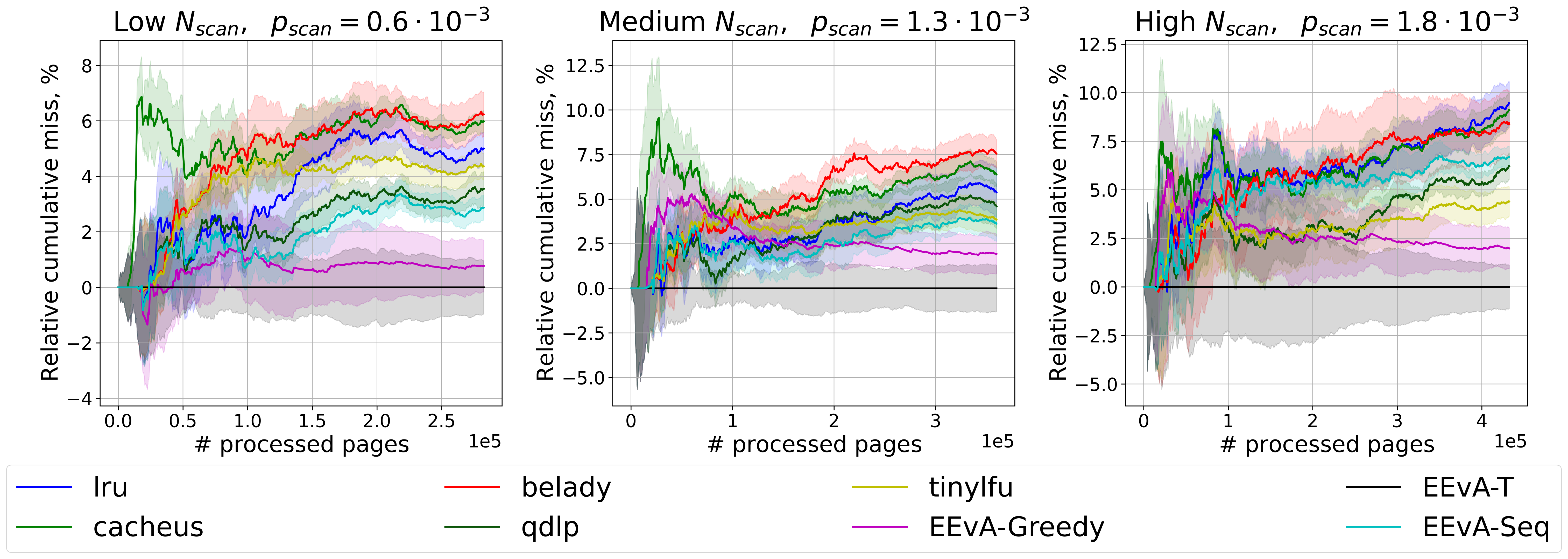}
    \caption{Comparison of the dynamics of the relative cumulative miss rate for the considered algorithms. EEvA-T gives the smallest values, therefore we consider it as the baseline. The reported averaged values and the corresponding half of variance are computed over 5 runs. Non-EEvA-based algorithms show significantly larger cumulative miss rates than our methods in all regular scenarios.}
    \label{fig::miss_dynamics}
\end{figure}

\subsection{Ablation study}
\label{sec: abl}

Here, we present an ablation study of the considered algorithms for different values of $p_{scan}$ parameter.
The buffer size in these tests is 10\% of the database size, $c_{scan} = 0.9$, $c_{get} = 1$, and the number of queries is $N = 5\cdot 10^5$.
If the value of $c_{scan}$ is much smaller than $c_{get}$, then the performance of EEvA-based algorithms is much better. 
Figure~\ref{fig::ablation_p_scan} shows the mean and variance for the miss rate and averaged time cost metrics in aggregation over 5 runs.
Note that EEvA-based algorithms show the best quality among competitors \emph{uniformly} for the considered range of $p_{scan}$ values.
This plot confirms the robustness of the proposed methods to the hyperparameter responsible for the ratio of scan queries.

In addition, we observe that our methods show the best performance wrt both metrics for $p_{scan} = 0$.
This phenomenon is clear because only \texttt{get}-type queries are simulated, and the probability of sampled pages from \texttt{scan}-accessed tables is 10 times smaller.
Thus, we artificially decrease the number of pages sampled with a high probability compared to pages from  \texttt{scan}-accessed tables.
Both metrics increase to maximum values after a smooth increase of $p_{scan}$.
This region of $p_{scan}$ corresponds to scenarios in which \texttt{scan}-type queries accidentally appear in traces and can not be properly managed by the considered algorithms.
After that, the number of \texttt{scan}-type queries becomes sufficiently large, and the algorithms can track them properly for efficient buffer utilization. 

\begin{figure}[!ht]
    \centering
        \includegraphics[width=0.5\linewidth]{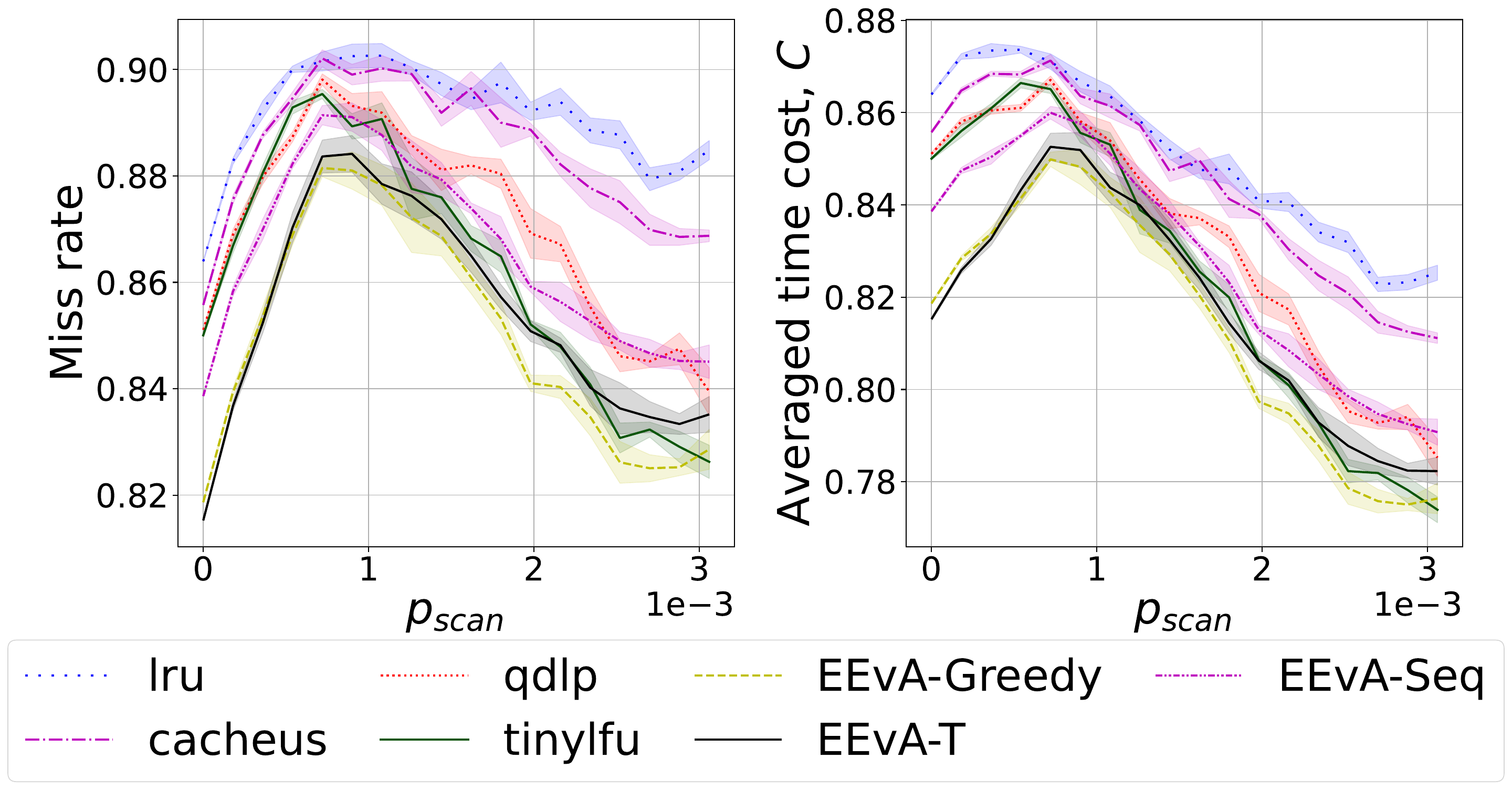}
    \caption{EEvA-based algorithms show the lowest miss rate and averaged time cost uniformly for the considered range of $p_{scan}$ values. The presented mean and variance values are computed over 5 runs.}
    \label{fig::ablation_p_scan}
\end{figure}

\subsection{Evaluation in DBMS}
\label{sec::opengauss}
To experimentally evaluate our approach in a DB system, we have implemented \texttt{EEvA-Seq} strategy in the buffer manager of an open-source openGauss~\cite{openGauss} database. 
We have benchmarked the proposed strategy on TPC-C~\cite{tpc-c} against the default page replacement strategy implemented in openGauss that we use as a baseline. 
We now briefly describe the implementation of the baseline and then present our experimental setting.

The main ideas behind the state-of-the-art implementation of buffer manager in the openGauss database are as follows. 
The buffer memory space is divided into $k$ parts, where $k$ is a configuration parameter. 
For each part, a set of candidate pages for eviction is maintained as a free list with elements referencing buffer pages. 
The free list is maintained by a worker thread that runs over the corresponding buffer part and pushes to the free list pages that are neither dirty, nor pinned. 
When a page needs to be loaded into the buffer, the first page referenced in the free list is taken as the victim. 
Then, the corresponding first element is removed from the free list. 
If all free lists are empty (which can be the case for heavy DB workloads), a victim page is directly searched in the buffer pool. 
The search begins with an element pointed by a special clock cursor that is also used for thread synchronization. 
Upon visiting the next page, the cursor moves one step ahead and is rewinded when the end of the buffer memory fragment is reached. 
The first non-dirty and non-pinned page found this way is taken as the victim. 

According to this logic, the decision on a page for replacement has a randomness factor introduced by read/write operations.
In contrast to the variants of frequency-based replacement strategies (e.g., $k$-bit clock algorithms adopted in many databases), the free list implementation in openGauss can be more efficient for DB application scenarios with large buffer size and random data access patterns (e.g., those exhibited by TPC-C queries). 
On the other hand, such a strategy may not be efficient in serving workloads with shifted data access patterns, including recurrent access to popular tables or ranges.

TPC-C benchmark includes five types of transactions: \texttt{newOrder}, \texttt{payment}, \texttt{orderStatus}, \texttt{delivery}, and \texttt{stockLevel}. 
Of these transactions, two types (\texttt{orderStatus} and \texttt{stockLevel}) are read-only transactions, and the remaining three types combine write and update operations. 
Importantly, one can change the distribution of these transactions by modifying benchmark parameters. 
In our experiments, we consider three settings: base (default TPC-C configuration), read-only, and write/update as presented in Table~\ref{tab::TPCC-Scenarios}.


\begin{table}[!ht]
\centering
\caption{Distribution of transaction types in TPC-C benchmark scenarios used in experiments.}
\label{tab::TPCC-Scenarios}
    \begin{tabular}{cccc}
    \toprule
     Transaction type & base & read-only & write/update \\
    \midrule
    \texttt{newOrder} & 45\% & 0\% & 49\% \\
\texttt{payment} & 43\%	& 0\% & 47\% \\
\texttt{orderStatus} & 4\% & 50\% & 0\% \\
\texttt{delivery} & 4\% & 0\% & 4\% \\
\texttt{stockLevel} & 4\% & 50 \% & 0\% \\
    \bottomrule
    \end{tabular}
\end{table}

For each scenario, we perform 20 runs of TPC-C benchmark: 10 runs with the default openGauss buffer management strategy and another 10 runs with \texttt{EEvA-Sec} strategy. 
The duration of each run is 5 minutes. 
For each run, we have measured the following parameters: \textit{transaction count} (the total number of completed transactions),  \textit{hit count} (the total number of buffer hits), and \textit{hits per transaction} (which is the number of hits divided by the number of transactions). 
The averaged values of these parameters over all runs are given in Table~\ref{tab::openGauss-Results}.

\begin{table}[!ht]
\centering
\caption{Evaluation results in three TPC-C scenarios. The proposed eviction strategy \texttt{EEvA-Seq} improves performance in all the considered scenarios.}
\label{tab::openGauss-Results}
    \begin{tabular}{ccccc}
    \toprule
     Scenario & Parameter &	Baseline &	\texttt{EEvA-Sec} &	Boost \% \\
    \midrule
   \multirow{3}{*}{base} & Transaction count & 57214 &	59426 &	+3,87 \\
   & Hit count &	14434992 &	15091925 & +4 \\
   & Hit per transaction &	252,29 &	253,96 &	+0,66 \\ \hline

\multirow{3}{*}{read} & Transaction count & 572000 & 711414 & +24,37 \\
   & Hit count &	148108032 &	206405162 & +39 \\
   & Hit per transaction &	258,93	& 290,15 &	+12,06 \\ \hline

\multirow{3}{*}{write/update} & Transaction count & 54588 & 55877 &	+2,36 \\
   & Hit count &	13107804 &	13424229 & +2 \\
   & Hit per transaction &	240,12 &	240,26	& +0,01 \\ 
    \bottomrule
    \end{tabular}
\end{table}

Results in Table~\ref{tab::openGauss-Results} show that \texttt{EEvA-Seq} strategy gives a higher transaction and hit count in all three scenarios. 
However, the most significant performance improvement is achieved in the read-only scenario.
The less significant improvement in the write/update scenario can be explained as follows. 
In bulk page writes or updates, most of the buffer becomes occupied by dirty pages, which cannot be replaced since they contain modified data and are waiting to be flushed to disk. 
Therefore, in this situation, searching for a candidate page for replacement becomes more difficult, and the impact of the computational overhead introduced by \texttt{EEvA-Seq}, compared to the baseline strategy, becomes more pronounced. 
Secondly, in massive updates, data migrates to new locations because the updated data is written to new pages. 
This effect introduces significant complexity and uncertainty for calculating pages' popularity based on the statistics of requests to these pages. 
We think that taking these effects into account can facilitate the development of even more efficient replacement strategies, which is one of the subjects of our future work.

This experimental evaluation shows that adaptive page replacement strategies can significantly increase the performance of the buffer manager in scenarios where read operations dominate over write/update operations. 
At the same time, optimizing page replacement policies alone is not enough to achieve optimal buffer manager performance. 
Smart policies for page flushing remain an important ingredient for optimal bugger management.

\section{Conclusion}
In this work, we have proposed a family of expert-based page replacement strategies for a DB buffer manager that employs experts of various granularity and differentiates between index and sequential access in updating expert rewards. 
We have provided an extensive experimental comparison of our strategies on synthetic benchmarks and showed that they outperform competitors (including SOTA recency/frequency-based strategies and ensemble strategy) in terms of the buffer hit ratio and averaged time costs. We have implemented the most promising strategy in the kernel of openGauss DB in which the logic of page replacement is tuned for DB application scenarios with a large buffer pool and heavy concurrent workload (such as TPC-C). Experimental evaluation showed that our strategy outperforms the baseline page replacement strategy of openGauss on TPC-C. Future work will include a deeper theoretical analysis of our strategies and improvements for practical implementation in a buffer manager, which takes into account page-flushing logic.   



\bibliographystyle{unsrt}
\bibliography{biblio}

\end{document}